\documentclass[a4paper,UKenglish,cleveref, autoref, thm-restate]{lipics-v2021}

\hideLIPIcs  

\title{Treedepth is NP-Hard on Planar Graphs of Maximum Degree 3}

\author{Kolja Kühn}{Institute of Theoretical Informatics, Karlsruhe Institute of Technology, Germany}{kolja.kuehn@kit.edu}{https://orcid.org/0009-0006-0746-6876}{}

\authorrunning{K. Kühn} 

\Copyright{Kolja Kühn} 

\ccsdesc{Theory of computation~Problems, reductions and completeness} 
\ccsdesc{Theory of computation~Graph algorithms analysis}
\ccsdesc{Mathematics of computing~Graph theory}

\keywords{Treedepth, planar graphs, NP-completeness, vertex ranking, elimination trees} 

\category{} 

\relatedversion{} 

\funding{This work was supported by KiKIT — the Helmholtz pilot program for core informatics research at the KIT.}

\nolinenumbers 

\EventEditors{John Q. Open and Joan R. Access}
\EventNoEds{2}
\EventLongTitle{42nd Conference on Very Important Topics (CVIT 2016)}
\EventShortTitle{CVIT 2016}
\EventAcronym{CVIT}
\EventYear{2016}
\EventDate{December 24--27, 2016}
\EventLocation{Little Whinging, United Kingdom}
\EventLogo{}
\SeriesVolume{42}
\ArticleNo{23}

\usepackage[disable]{todonotes}
\usepackage{thm-restate}
\usepackage{complexity}
\usepackage{amsmath}
\usepackage{tikz}
\usepackage{xspace}

\begin{document}

\usetikzlibrary{calc}
\definecolor{edgered}{HTML}{D9534F}
\definecolor{edgeblue}{HTML}{3B82C4}
\definecolor{vfill}{HTML}{F1EFEA}
\definecolor{vredfill}{HTML}{FBE3E1}

\newcommand{\wallcoords}[6]{%
  \foreach \p [count=\i from 0] in {#2} \foreach \n [count=\t from 0] in {#3}
    \coordinate (#1-\p\n) at ($(#4)+\t*(#5)+\i*(#6)$);}
\newcommand{\drawedges}[3]{\edef\edgelist{#3}\foreach \u/\v in \edgelist \draw[#2] (#1-\u) -- (#1-\v);}
\newcommand{\drawvertices}[5]{\foreach \p in {#3} \foreach \n in {#4} \node[#2] at (#1-\p\n) {#5};}
\newcommand{\stairwall}[6]{%
  \foreach \c in {1,...,#3}{%
    \draw[#6] ({#1+(\c-1)*#4},{#2+(\c-1)*#5}) -- ++(0,{(2*#3-1)*#5});
    \ifnum\c<#3 \foreach \j in {1,...,#3}
      \draw[#6] ({#1+(\c-1)*#4},{#2+(2*\j+\c-2)*#5}) -- ++(#4,0); \fi}}

\definecolor{bsred}{HTML}{E3726E}
\definecolor{bsblue}{HTML}{5B9BD5}
\definecolor{bsmix}{HTML}{A985B6}    
\def\bswidth{6.6mm}
\def\bsopacity{0.33}
\newcommand{\bset}[4]{%
  \begin{scope}[transparency group, opacity=\bsopacity]
    \foreach \u/\v in {#3} \draw[#2, line width=\bswidth, line cap=round] (#1-\u) -- (#1-\v);
    \foreach \x in {#4} \fill[#2] (#1-\x) circle[radius=0.5*\bswidth];
  \end{scope}}

\tikzset{
  minor/.style={
    vertex/.style={circle, draw=black!70, fill=vfill, inner sep=0pt, minimum size=4.8mm, font=\scriptsize},
    term/.style={vertex, line width=1.1pt, draw=black},
    ghost/.style={vertex, dashed, draw=black!65, fill=white, text=black!65},
    vdel/.style={ghost, densely dashed, line width=0.9pt, draw=black},
    rminor/.style={vertex, fill=bsred!45, draw=edgered!80!black},
    bminor/.style={vertex, fill=bsblue!45, draw=edgeblue!80!black},
    mminor/.style={vertex, fill=bsmix!55, draw=bsmix!60!black},
    keep/.style={black!80, thick},
    red/.style={edgered, very thick},
    blue/.style={edgeblue, very thick},
    mix/.style={bsmix!75!black, very thick},
    del/.style={black!45, thick, dash pattern=on 3pt off 2.5pt}},
  strands/.style={
    dot/.style={circle,fill,inner sep=1.1pt},
    sdot/.style={circle,fill=teal!70!black,inner sep=1.7pt},
    rdot/.style={circle,draw,fill=white,inner sep=1.4pt},
    tnode/.style={circle,draw,inner sep=1.5pt,minimum size=17pt}}}

\newcommand{\hwallcoords}[3]{\wallcoords{#1}{#2}{#3}{-0.75,0}{0.75,0}{0.75,-1.1}}
\def\WFourEdges{a1/a2,a2/a3,a3/a4,a4/a5,a5/a6,a6/a7,a7/a8,%
  b1/b2,b2/b3,b3/b4,b4/b5,b5/b6,b6/b7,b7/b8,%
  c1/c2,c2/c3,c3/c4,c4/c5,c5/c6,c6/c7,c7/c8,%
  d1/d2,d2/d3,d3/d4,d4/d5,d5/d6,d6/d7,d7/d8,%
  a2/b1,a4/b3,a6/b5,a8/b7,b2/c1,b4/c3,b6/c5,b8/c7,c2/d1,c4/d3,c6/d5,c8/d7}
\newcommand{\wallTikz}{%
    \begin{tikzpicture}[minor]
      \hwallcoords{A}{a,b,c,d}{1,...,8}
      \drawedges{A}{thick, black!75}{\WFourEdges}
      \drawvertices{A}{vertex}{a,b,c,d}{1,...,8}{$\p_{\n}$}
    \end{tikzpicture}%
}

\def\FigTwoDel{a6/b5,b2/b3,b4/b5,b5/b6,b6/b7,b6/c5,c6/d5}
\def\FigTwoRed{a5/a6,a6/a7,c5/c6,c6/c7,d5/d6,d6/d7}
\def\FigTwoBlue{a2/b1,b1/b2,a4/b3,b3/b4,a8/b7,b7/b8}
\def\FigTwoKeep{a1/a2,a2/a3,a3/a4,a4/a5,a7/a8,c1/c2,c2/c3,c3/c4,c4/c5,c7/c8,%
  d1/d2,d2/d3,d3/d4,d4/d5,d7/d8,b2/c1,b4/c3,b8/c7,c2/d1,c4/d3,c8/d7}
\def\WThreeMinor{a1/a2,a2/a3,a3/a4,a4/a7,a7/a8,c1/c2,c2/c3,c3/c4,c4/c7,c7/c8,%
  d1/d2,d2/d3,d3/d4,d4/d7,d7/d8,a2/c1,a4/c3,a8/c7,c2/d1,c4/d3,c8/d7}

\newcommand{\wallMinor}{%
    \begin{tikzpicture}[minor]
      \hwallcoords{B}{a,b,c,d}{1,...,8}
      \bset{B}{bsred}{a5/a6,a6/a7}{a5,a6,a7}
      \bset{B}{bsred}{c5/c6,c6/c7}{c5,c6,c7}
      \bset{B}{bsred}{d5/d6,d6/d7}{d5,d6,d7}
      \bset{B}{bsblue}{a2/b1,b1/b2}{a2,b1,b2}
      \bset{B}{bsblue}{a4/b3,b3/b4}{a4,b3,b4}
      \bset{B}{bsblue}{a8/b7,b7/b8}{a8,b7,b8}
      \drawedges{B}{del}{\FigTwoDel}
      \drawedges{B}{keep}{\FigTwoKeep}
      \drawedges{B}{red}{\FigTwoRed}
      \drawedges{B}{blue}{\FigTwoBlue}
      \drawvertices{B}{vertex}{a,c,d}{1,...,7}{$\p_{\n}$}
      \drawvertices{B}{vertex}{b}{1,2,3,4,7}{$\p_{\n}$}
      \drawvertices{B}{term}{a,b,c,d}{8}{$\p_{\n}$}
      \node[ghost] at (B-b6) {$b_6$};
      \node[vdel] (v) at (B-b5) {$b_5$};
      \node[font=\small] at (3.0,-3.95) {$W_4 - v$};
      \draw[-stealth, very thick, black!75] (7.35,-1.65) -- ++(0.9,0);
      \begin{scope}[xshift=8.9cm]
        \hwallcoords{C}{a,c,d}{1,2,3,4,7,8}
        \drawedges{C}{keep}{\WThreeMinor}
        \drawvertices{C}{vertex}{a,c,d}{1,3}{$\p_{\n}$}
        \drawvertices{C}{vertex}{c,d}{2,4}{$\p_{\n}$}
        \drawvertices{C}{bminor}{a}{2,4}{$\p_{\n}$}
        \drawvertices{C}{rminor}{a,c,d}{7}{$\p_{\n}$}
        \node[bminor, line width=1.1pt, draw=edgeblue!60!black] at (C-a8) {$a_8$};
        \drawvertices{C}{term}{c,d}{8}{$\p_{\n}$}
        \node[font=\small] at (2.25,-3.95) {$W_3$};
      \end{scope}
    \end{tikzpicture}%
}

\newcommand{\cfAnchor}[3]{
  \draw (#1,#2) -- (#1,#3) -- (#1-0.035,#3+0.12);
  \draw (#1,#3) -- (#1+0.035,#3+0.12);
  \foreach \s in {-1,1}
    {\stairwall{#1+\s*0.035}{#3+0.12}{8}{\s*0.026}{0.013}{gray!70!black,line width=0.15pt}}
  \node[rdot] at (#1,#3) {};
  \node[dot] at (#1,#2) {};}
\newcommand{\cfStrand}[5]{
  \draw[teal!70!black,thick] (#1,#2) |- ({(#1+#3)/2},#5) -| (#3,#4);
  \node[dot] at (#1,#2) {}; \node[dot] at (#3,#4) {};
  \node[sdot] at ({(#1+#3)/2},#5) {};}
\newcommand{\constructionFigure}{%
    \begin{tikzpicture}[x=1cm,y=1cm,strands]
      \node[tnode] (u) at (1.0,4.6) {$u$};
      \node[tnode] (v) at (4.95,4.6) {$v$};
      \node[tnode] (w) at (9.45,3.95) {$w$};
      \node[tnode] (z) at (12.8,5.25) {$z$};
      \draw (u) -- node[above]{$c=1$} (v);
      \draw (v) -- node[below,sloped]{$c=2$} (w);
      \draw (v) -- node[above,sloped]{$c=1$} (z);
      \node[anchor=west] at (-0.2,5.7) {$K=3 \rightsquigarrow K+1=4$};
      \foreach \x/\k/\lx/\lbl in {0/5/1.0/u, 3.2/8/4.95/v, 8.2/6/9.45/w, 11.8/5/12.8/z} {
        \stairwall{\x}{-(2*\k-1)*0.25}{\k}{0.5}{0.25}{gray!75}
        \node at (\lx,{-(2*\k-1)*0.25-0.4}) {$W_\lbl\cong W_\k$};
      }
      \foreach \x/\y/\r in {0/0/1.25, 3.7/0.25/2.0, 9.2/0.5/1.5, 12.3/0.25/1.25}
        \foreach \q in {0,...,3}
          {\cfAnchor{\x+\q*0.5}{\y+\q*0.25}{\r}}
      \cfStrand{2.0}{1.0}{3.2}{0}{1.9}       
      \cfStrand{6.2}{1.5}{8.7}{0.25}{2.75}   
      \cfStrand{6.7}{1.75}{8.2}{0}{2.4}      
      \cfStrand{5.7}{1.25}{11.8}{0}{3.1}     
      \node[sdot,label=right:{strand vertex}] at (1.2,-4.85) {};
      \node[rdot,label=right:{anchor root $r$}] at (5.0,-4.85) {};
      \stairwall{8.95}{-4.85}{8}{0.026}{0.013}{gray!70!black,line width=0.15pt}
      \node[right] at (9.15,-4.85) {copy of $W_a=W_8$};
    \end{tikzpicture}%
}

\def\WThreeDel{a6/b5,b2/b3,b4/b5,b5/b6,b6/c5}
\def\WThreeRed{a2/b1,b1/b2,b2/c1,a3/a4,a4/a5,a5/a6,a4/b3,b3/b4,b4/c3,c3/c4,c4/c5,c5/c6}
\def\WThreeBlue{a1/a2,a2/a3,c1/c2,c2/c3}
\newcommand{\doublingGraph}{%
    \begin{tikzpicture}[
        wv/.style={circle, draw=black!70, fill=white, inner sep=0pt, minimum size=2.4mm},
        rv/.style={circle, draw=black!70, fill=white, inner sep=0pt, minimum size=5mm},
        we/.style={thick, black!75},
        wall/.pic={
          \wallcoords{}{a,b,c}{1,...,6}{0,0}{0,-\dy}{\dx,-\dy}
          \drawedges{}{we}{\WThreeDel,\WThreeRed,\WThreeBlue}
          \drawvertices{}{wv}{a,b,c}{1,...,6}{}
        }]
        \def\dx{0.9}    
        \def\dy{0.65}   
        \def\gap{1.3}   
        \def\sep{0.8}   
        \def\alen{0.9}  
        \pgfmathsetmacro{\H}{7*\dy}                    
        \pgfmathsetmacro{\Lx}{\H+2*\sep+\alen+2*\dx}  
        \pgfmathsetmacro{\ay}{0.5*\H}                  
        \pic[cm={0,-1,-1,0,(0,{\ay+\dx})}] at (0,0) {wall};
        \pic[xscale=-1] (L) at (\Lx,\H) {wall};
        \pic (R) at (\Lx+\gap,\H) {wall};
        \coordinate (r) at (\Lx+0.5*\gap,\H+1.1);
        \draw[we] (L-a1) -- (r) -- (R-a1);
        \node[wv] at (L-a1) {};
        \node[wv] at (R-a1) {};
        \node[rv, font=\bfseries] at (r) {r};
        \draw[-stealth, very thick, black!75] ({\H+\sep},\ay) -- ++(\alen,0);
    \end{tikzpicture}%
}

\def\FigFiveRed{c4/c5,c5/c6}
\def\FigFiveBlue{a2/b1,b1/b2}
\def\FigFiveMix{a4/a5,a5/a6,a4/b3,b3/b4}
\def\FigFiveKeep{a1/a2,a2/a3,a3/a4,c1/c2,c2/c3,c3/c4,b2/c1,b4/c3}
\def\WTwoMinor{a1/a2,a2/a3,a3/a6,c1/c2,c2/c3,c3/c6,a2/c1,a6/c3}
\newcommand{\csTrail}[1]{
  \draw[thick, black!75] (#1) -- ++(0,0.3);
  \draw[line width=1.2pt, line cap=round, dash pattern=on 0pt off 2.8pt]
    ($(#1)+(0,0.4)$) -- ++(0,0.3);}
\newcommand{\strandCaseFigure}{%
    \begin{tikzpicture}[x=2cm,y=2cm,minor,
      vertex/.append style={minimum size=2.4mm},
      svertex/.style={circle, fill=teal!70!black, inner sep=0pt, minimum size=3.2mm},
      strand/.style={teal!70!black, very thick},
      sdel/.style={strand, thick, dash pattern=on 3pt off 2.5pt}]
      \def\bswidth{4.4mm}
      \foreach \w/\o in {U/0, V/2.5} {\wallcoords{\w}{a,b,c}{1,...,6}{\o,-1.25}{0,0.25}{0.5,0.25}}
      \foreach \w in {U,V} {
        \bset{\w}{bsmix}{a4/a5,a5/a6,a4/b3,b3/b4}{a4,a5,a6,b3,b4}
        \bset{\w}{bsred}{c4/c5,c5/c6}{c4,c5,c6}
        \bset{\w}{bsblue}{a2/b1,b1/b2}{a2,b1,b2}
      }
      \draw[strand] (U-c6) |- (1.75,0.9) -| (V-a6);
      \draw[sdel] (U-b6) |- (1.75,1.3) -| (V-b6);
      \foreach \w in {U,V} {
        \drawedges{\w}{del}{\WThreeDel}
        \drawedges{\w}{keep}{\FigFiveKeep}
        \drawedges{\w}{red}{\FigFiveRed}
        \drawedges{\w}{blue}{\FigFiveBlue}
        \drawedges{\w}{mix}{\FigFiveMix}
      }
      \csTrail{U-a6}
      \csTrail{V-c6}
      \foreach \w in {U,V} {
        \drawvertices{\w}{vertex}{a,c}{1,2,3,4,5}{}
        \drawvertices{\w}{term}{a,c}{6}{}
        \drawvertices{\w}{vertex}{b}{1,2,3,4}{}
        \drawvertices{\w}{ghost}{b}{5,6}{}
      }
      \node[svertex] at (1.75,0.9) {};
      \node[ghost, minimum size=3.2mm, label=above:{$x$}] at (1.75,1.3) {};
      \node[left=2pt] at (U-b6) {$x'$};
      \node[left=2pt] at (V-b6) {$x''$};
      \node at (0.5,-1.55) {$W_u\cong W_3$};
      \node at (3.0,-1.55) {$W_v\cong W_3$};
      \draw[-stealth, very thick, black!75] (4.0,-0.2) -- (4.5,-0.2);
      \begin{scope}[xshift=10.2cm]
        \foreach \w/\o in {U/0, V/1.5} {\wallcoords{R\w}{a,c}{1,2,3,6}{\o,-0.75}{0,0.25}{0.5,0.25}}
        \draw[strand] (RU-c6) |- (1.0,0.75) -| (RV-a6);
        \foreach \w in {RU,RV} {\drawedges{\w}{keep}{\WTwoMinor}}
        \csTrail{RU-a6}
        \csTrail{RV-c6}
        \foreach \w in {RU,RV} {
          \drawvertices{\w}{vertex}{a,c}{1,3}{}
          \drawvertices{\w}{vertex}{c}{2}{}
          \drawvertices{\w}{bminor}{a}{2}{}
          \drawvertices{\w}{mminor, line width=1.1pt, draw=bsmix!45!black}{a}{6}{}
          \drawvertices{\w}{rminor, line width=1.1pt, draw=edgered!60!black}{c}{6}{}
        }
        \node[svertex] at (1.0,0.75) {};
        \node at (0.25,-1.05) {$W'_u\cong W_2$};
        \node at (1.75,-1.05) {$W'_v\cong W_2$};
      \end{scope}
    \end{tikzpicture}%
}
\newcommand{\wts}{\textsc{Weighted Tree Search}\xspace}
\newcommand{\tdp}{\textsc{Treedepth on Planar Graphs of Degree at most 3}\xspace}
\newcommand{\td}{\operatorname{td}}
\newcommand{\double}{\operatorname{Dbl}}
\newcommand{\wall}{W}

\maketitle

\begin{abstract}
Treedepth is an important parameter in algorithmic graph theory.
It is known that the problem of computing the treedepth of a graph is \NP-hard.
It has been an open problem whether treedepth can be computed in polynomial time when the input graph is restricted to be planar, or of bounded degree.
Assuming $\P \neq \NP,$ we answer both questions in the negative:
Under polynomial-time Turing reductions, it is \NP-hard to compute the treedepth of planar graphs with maximum degree 3.
This degree bound is best possible, since treedepth can be computed efficiently on graphs of maximum degree 2.

\end{abstract}

\subsection*{AI Disclosure}
The idea for the reduction was suggested by Claude Opus 5.5.
Our contribution was to improve the degree bound from 4 to 3 and to simplify the lower bound argument.
The text was written by the author, with AI tools used only for figures and proofreading.
The author (of course) takes full responsibility for the content of this paper.

\section{Introduction}
\emph{Treedepth} is a graph parameter that plays a central role in both algorithmic and structural graph theory.
Intuitively, it measures how efficiently a graph can be deconstructed using vertex deletions.
Consider for example a star graph $S$ with central vertex $v$.
The graph $S - v$ is an independent set, i.e., contains no large connected components.
Indeed, a star graph has small treedepth and the parameter treedepth has even been described as a measure of how similar a graph is to a star graph.

Graphs with bounded treedepth allow for efficient dynamic programming approaches.
Many problems that are \NP-hard in general can be solved efficiently when restricted to graphs of bounded treedepth.
For example, \textsc{3-Coloring}, \textsc{Dominating Set} and \textsc{Vertex Cover} can be solved in time $2^{\mathcal{O}(d)} \cdot n^{\mathcal{O}(1)}$ and space $\mathcal{O}(d \cdot \log n)$ for graphs on $n$ vertices with treedepth at most $d$~\cite{treedepth-space-efficient}.
However, these algorithms require a treedepth decomposition of the graph as input.
Constructing such an optimal decomposition, or even just computing the treedepth of a graph, are problems which have long been known to be \NP-hard~\cite{pothen1988complexity}.
They remain \NP-hard on bipartite graphs~\cite{bodlaender1998rankings}, co-bipartite graphs~\cite{pothen1988complexity} and chordal graphs~\cite{dereniowski2006vertex}.
On the positive side, treedepth can be computed in linear time on trees~\cite{schaffer1989optimal}.
It has been conjectured that the treedepth problem remains \NP-hard on planar graphs and on graphs of bounded degree~\cite{schirrmacher2025elimination}.
Some progress towards these conjectures was made by showing that the weighted variant of treedepth is \NP-hard on trees~\cite{dereniowski2006vertex} and on graphs of maximum degree 10~\cite{dirks2026weightedtreedepthnpcompletegraphs}.
We resolve both questions with the following theorem.
\begin{restatable}{theorem}{hardnessTheorem}
    \tdp is \NP-hard under polynomial-time Turing reductions.
    In particular, unless $\P = \NP$, there is no polynomial-time algorithm that decides for a given planar graph $G$ with maximum degree 3 and an integer $k$, whether $\td(G) \leq k$.
\end{restatable}

This paper is organized as follows.
We give some basic definitions and simple observations in \cref{sec:definitions}.
In \cref{sec:construction} we describe the polynomial transformation used in the reduction.
\cref{sec:proof} shows the main result:
\tdp is \NP-hard by a reduction from \wts.
Finally, we give some open problems in \cref{sec:conclusion}.

\section{Definitions and Simple Observations}
\label{sec:definitions}
All graphs in this paper are simple and undirected.
We use standard terminology from graph theory, following Diestel~\cite{Diestel}.

An important graph for our reduction is the \emph{wall}.
It has large treedepth, a small maximum degree and provides robust connectivity.
For each $k \in \mathbb{N}_0$, we define the $k \times k$ wall $\wall_k$ as the union of $k$ paths $P_1, \dots, P_k$, where $P_i$ has vertices $(i,1), \dots, (i,2k)$, together with additional edges $(i,2j)(i + 1, 2j - 1)$ for $i \in [k - 1], j \in [k]$. See \cref{fig:wall} for an example.
\begin{figure}
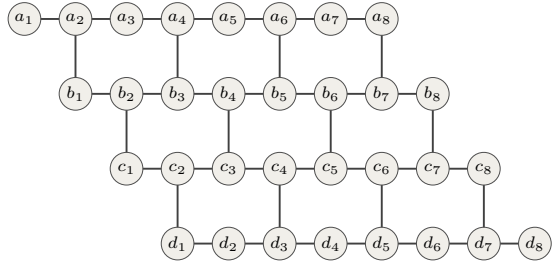

    \centering
    \scalebox{0.9}{
        \wallTikz
    }
    \caption{The $4 \times 4$ wall $\wall_4$. We label the paths by letters, so $b_1 = (2,1)$.}
    \label{fig:wall}
\end{figure}
We call the final vertex $(i,2k)$ of path $P_i$ its \emph{terminal}.
We note that the wall $\wall_k$ has maximum degree 3 and each of its $k$ terminals has degree at most 2.
Its robust connectivity can be seen in the following fact.
The graph obtained from $\wall_k$ by deleting any vertex $v \in V(\wall_k)$ contains as a minor the wall $\wall_{k - 1}$.
A stronger version of this property is crucial in our (inductive) hardness proof, so we formalize it in the following observation.
For an intuitive understanding, it suffices to inspect \cref{fig:wall-minor}.
\begin{observation}
\label{obs:prel:wall-contract}
    Let $v = (i,j)$ be any vertex of the wall $\wall_k$ for $k \geq 1$.
    Then $\wall_k - v$ contains a minor model of $\wall_{k - 1}$ with the following property.
    For each $i' < i$, terminal $(i',2k)$ of $\wall_k$ is contained in the branch set for terminal $(i', 2k- 2)$ of $\wall_{k - 1}$, and for each $i' > i$, terminal $(i',2k)$ is contained in the branch set for terminal $(i' - 1, 2k-2)$.
    In particular, there is a bijection between the two sets of terminals (excluding terminal $(i,2k)$).
\end{observation}
\begin{figure}
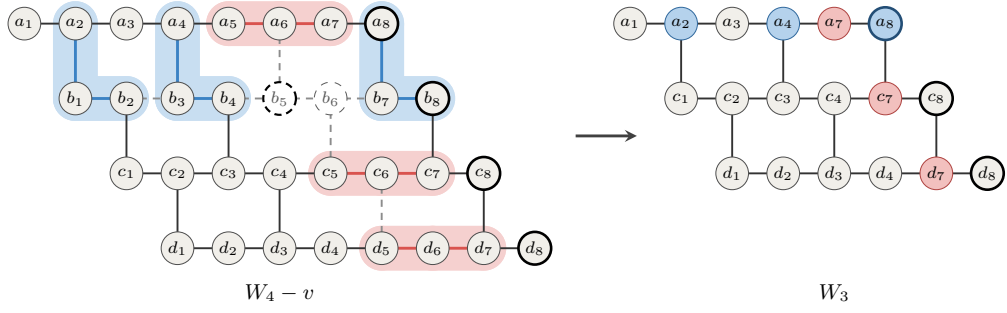

    \centering
    \scalebox{0.9}{
        \wallMinor
    }
    \caption{A minor model of $\wall_3$ in $\wall_4 - v$ for $v = b_5 = (2,5)$, as in \cref{obs:prel:wall-contract}. Dashed vertices and edges are deleted. Red are branch sets of the first family, blue are branch sets of the second family, remaining vertices are singleton branch sets. Terminals are drawn with thick vertex boundaries.}
    \label{fig:wall-minor}
\end{figure}
\begin{proof}
    The statement trivially holds for $k = 1$ so assume $k \geq 2$ for the remainder of the proof.
    We first define a subgraph $H$ of $\wall_k - v$, then show that $\wall_{k-1}$ is a minor of $H$ with the desired properties.
    For this, first observe that $v$ is one of $(i,2\ell -1), (i,2\ell)$ for $\ell = \lceil \frac{j}{2} \rceil$.
    Define $H$ as the graph obtained from $\wall_k$ by deleting the vertices $(i, 2\ell - 1), (i, 2\ell)$ and edges of the form $(i, 2m)(i,2m + 1)$ for $m \in [k - 1]$, and $(i', 2\ell)(i' + 1, 2\ell -1)$ for $i' \in [k - 1]$.
    See \cref{fig:wall-minor}, with deleted edges and vertices in dashed gray.

    For the second step, we first consider the main case where $1 < i < k$ and $\ell < k$.
    Then define in $H$ the following three families of branch sets:
    \begin{enumerate}
        \item For each $i' \neq i$ the branch set $\{(i', 2\ell - 1), (i',2\ell), (i', 2\ell +1)\}$, highlighted in red in \cref{fig:wall-minor}.
        \item For each $m \neq \ell$ the branch set $\{(i - 1, 2m), (i, 2m-1), (i,2m)\}$ highlighted in blue in \cref{fig:wall-minor}.
        \item For each remaining vertex $x$, a singleton branch set $\{x\}$.
    \end{enumerate}

    We show that contracting the branch sets in $H$ yields $\wall_{k - 1}$ with the desired properties.
    Each remaining path is now a path on $2k - 2$ vertices.
    All edges between singleton branch sets remain unchanged.
    The connecting edges between former vertices $(i - 1, 2m)$ and $(i + 1, 2m-1)$ are given by $(i + 1, 2m- 1)$ being connected to $(i, 2m)$ in $H$, and $(i,2m)$ being in a branch set with $(i - 1,2m)$.
    Finally, the terminal condition holds by construction.

    If $i = 1$ or $i = k$, the second family of branch sets is omitted and path $P_i$ is simply deleted instead.
    
    If $\ell = k$, the first family of branch sets is instead defined as $\{(i', 2\ell - 1), (i',2\ell), (i', 2\ell - 2)\}$ for $i' \neq i$.
    If additionally $1 < i < k$, then two of the branch sets overlap, namely $\{(i - 1, 2k - 1), (i - 1,2k), (i - 1, 2k - 2)\}$ and $\{(i - 1, 2[k - 1]), (i, 2[k-1] - 1), (i,2[k - 1])\}$.
    We merge them into one branch set.
    This set is connected, contains the terminal $(i-1,2k)$, and has as neighbors in $H$ the vertices $(i-1, 2k-3)$ and $(i + 1, 2k-3)$ as required for the resulting terminal of row $i - 1$.
    In each case, the remaining argument is analogous.
\end{proof}

We now define the main protagonists of this paper, the \emph{treedepth} and the \emph{weighted tree search} problems.

\paragraph*{Treedepth}
\emph{Treedepth} ($\td$ for short) is a graph parameter which is defined recursively.
\begin{definition}[Treedepth]
\label{def:treedepth}
    The treedepth of the graph without vertices is $\td(\emptyset) := 0$ and the treedepth of the single-vertex graph $K_1$ is $\td(K_1) := 1$.
    The treedepth of the disjoint union of two graphs $G_1, G_2$ is $\td(G_1 \cup G_2) := \max\{\td(G_1), \td(G_2)\}$.
    The treedepth of a connected graph $G$ on more than one vertex is defined as $\td(G) := 1 + \min_{v \in V(G)}\{\td(G - v)\}$.
\end{definition}

It is easy to see that the treedepth of the wall $\wall_k$ is linear in $k$, in particular, between $k$ and $4k$.\footnote{Treedepth is bounded from below by treewidth, and $\operatorname{tw}(\wall_k) = k$.
The upper bound of $4k$ is obtained by recursively first deleting the middle vertex of each path, followed by deleting on each resulting connected component the middle path of each half-wall, to obtain a wall with half the dimensions.}
It is known that treedepth is monotone under taking minors, i.e., if $H$ is a minor of $G$, then $\td(H) \leq \td(G)$~\cite{Nešetřil2012}.
In particular, using \cref{obs:prel:wall-contract}, it holds that $\td(\wall_{k - 1}) \leq \td(\wall_k - v)$ for any vertex $v \in V(\wall_k)$.%
\footnote{This also gives a direct inductive proof of $\td(\wall_k) \geq k$.}

An equivalent definition for the treedepth of a connected graph is the following.
\begin{definition}[Nešetřil, Ossona de Mendez~\cite{NESETRIL20061022,Nešetřil2012}]
    An \emph{elimination tree} for a connected graph $G$ is a rooted tree $F$ with vertex set $V(G)$ defined recursively:
    If $V(G) = \{x\}$ then $F$ is just $\{x\}$.
    Otherwise, a vertex $r \in V(G)$ is chosen as the root of $F$ and the branches of $F$ at $r$ are elimination trees of the connected components of $G - r$, with each of their roots being connected to $r$ in $F$.

    The treedepth of a connected graph $G$ is the minimum vertex-height of an elimination tree for $G$.
    We call an elimination tree of this minimum height an \emph{optimal elimination tree}.
\end{definition}

An instance $I = (G,k)$ of the \tdp problem consists of a planar graph $G$ with maximum degree 3, and an integer $k$, and asks whether $\td(G) \leq k$.
Using the definition of treedepth via elimination trees, it is easy to see that the \tdp problem is in \NP.
To simplify the write-up, we use the two given definitions interchangeably when talking about treedepth.

For our construction we require as a gadget a graph of large treedepth whose optimal elimination tree is rooted at a specific vertex.
While the wall has large treedepth, it is not known which root vertex attains its optimal treedepth.
To work around this, we define a double wall for which we can pin down an optimal root vertex.
\begin{definition}
\label{def:preliminaries:doubling-graph}
    Let $\wall_k$ be the $k \times k$ wall.
    We define the double wall $\wall_k^{(2)}$ as the graph obtained by adding to the disjoint union of two copies of $\wall_k$ a subdivided edge between their copies of the vertex $(k,2k)$.
    We call this subdivision vertex $r$ the \emph{root} of the double wall.
    See \cref{fig:preliminaries:doubling-graph} for an example.
\end{definition}

\begin{figure}
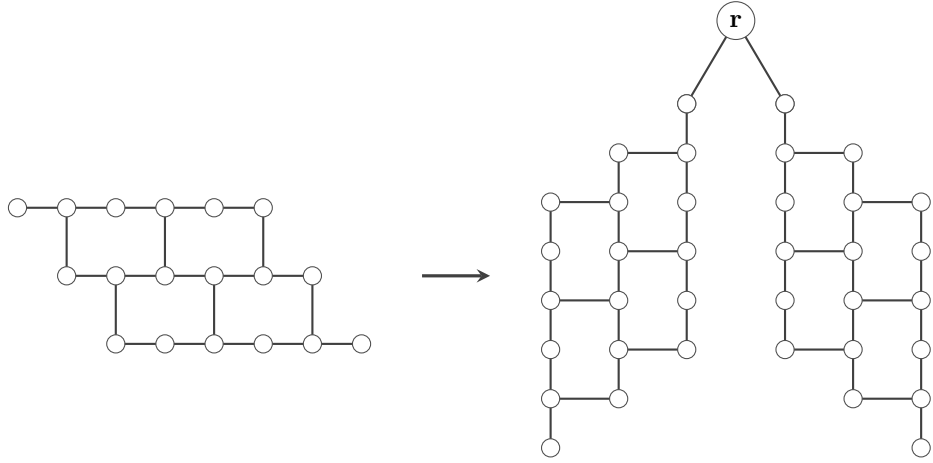

    \centering
    \doublingGraph
    \caption{Left: The wall $\wall_3$. Right: The double wall $\wall_3^{(2)}$ with root $r$. Note that $\wall_3^{(2)}$ has maximum degree 3.}
    \label{fig:preliminaries:doubling-graph}
\end{figure}

\begin{observation}
\label{obs:preliminaries:td-of-g-G}
    For every integer $k$ it holds that the double wall $\wall_k^{(2)}$ with root $r$ has treedepth $\td(\wall_k^{(2)}) = \td(\wall_k) + 1$.
    Furthermore, $\td(\wall_k^{(2)} - r) = \td(\wall_k)$, i.e., vertex $r$ is one that attains the minimum in the recursive definition of treedepth.
\end{observation}

\begin{proof}
    It holds by definition that $\td(\wall_k^{(2)}) = 1 + \min_{v \in V(\wall_k^{(2)})}\{\td(\wall_k^{(2)} - v)\}$.
    For any choice of $v$ it holds that $\wall_k^{(2)} - v$ contains as a subgraph a copy of $\wall_k$.
    Since treedepth is monotone under taking subgraphs, it follows that $\td(\wall_k^{(2)}) \geq \td(\wall_k) + 1$.
    To see equality, as well as the second claim, observe that $\wall_k^{(2)} - r$ consists of two disjoint copies of $\wall_k$.
    Thus, its treedepth is $\td(\wall_k \cup \wall_k) = \td(\wall_k)$ by definition.
    The claim follows.
\end{proof}

\cref{obs:preliminaries:td-of-g-G} highlights a useful intuition for treedepth and elimination trees.
The root $r$ is a cut vertex that splits the graph $\wall_k^{(2)}$ in a balanced way.
This allows the two connected components of $\wall_k^{(2)} - r$ to be deconstructed ``in parallel''.
Similarly, the optimal elimination tree for our construction (defined in \cref{sec:construction}) separates components of large treedepth to deconstruct them in parallel.
The problem then boils down to finding the most efficient way to separate these components.

\paragraph*{Weighted Tree Search}
A \emph{weighted} tree $(T,c)$ has a weight function $c : E(T) \rightarrow \mathbb{N}$ which assigns to each edge of the tree $T$ a positive integer weight.
We further define $c(T) := \sum_{e \in E(T)} c(e)$.
In the \wts problem, we are given a weighted tree $(T,c)$ and a budget $K$, and the goal is to find a secret vertex $v \in V(T)$ using the following operation.
We may query any edge $e$ of $T$ at cost $c(e)$, to reveal in which of the two connected components of $T - e$ the secret vertex $v$ lies.
The search ends when we have narrowed down $v$ to be one specific vertex.
The goal is to minimize the worst-case cost of the entire search.

We use the decision tree formalization defined in~\cite{CICALESE2012100}, where the problem is called the \emph{binary identification problem for weighted trees}.
A decision tree $D$ for a weighted tree $(T,c)$ is a binary tree.
We call the vertices of $D$ \emph{nodes} to disambiguate them from the vertices of $T$.
If $T$ has only one vertex $v$, then $D$ consists of a single node labeled $v$.
If $T$ has at least one edge, then $D$ has a root labeled by some edge $e = uv$ connected to two subtrees which are decision trees for $T_u$ and $T_v$, the connected components of $T - e$.
    
Observe that the definition of the decision tree mirrors that of an elimination tree, except that its inner nodes correspond to edges rather than vertices.%
\footnote{Put differently, removing all leaves of a decision tree $D$ for some tree $T$ yields an elimination tree of the line graph of $T$.}

We define the cost of a decision tree $D$ as 0, if it consists of just one isolated node.
Otherwise, it is rooted at some $T$-edge $e = uv$, and has cost $c(e) + \max\{\operatorname{cost}(D_u), \operatorname{cost}(D_v)\}$ for its subtrees $D_u$ and $D_v$.

For a weighted tree $(T,c)$ we define the search cost function $f$ as the minimum cost of a decision tree $D$ for $(T,c)$.
Equivalently, $f$ can be defined as $f(K_1,c) = 0$ and $f(T,c) = \min_{e \in E(T)}\Big\{c(e) + \max\{f(T_1, c|_{T_1}), f(T_2, c|_{T_2})\}\Big\},$ where $T_1, T_2$ are the connected components of $T - e$.
An instance $((T,c), K)$ of the \wts problem then asks whether $f(T,c) \leq K$.

\begin{theorem}[Cicalese, Jacobs, Laber and Valentim \cite{CICALESE2012100}]
    The \wts problem is strongly \NP-hard.
    That is, the problem remains \NP-hard when the parameter $K$ or, equivalently, the weight function $c$ is encoded in unary.%
    \footnote{For any tree $T$ on more than one vertex, if $K < \max_{e \in E(T)} \{c(e)\}$, then the instance $((T,c),K)$ is trivially a No-instance, while if $K \geq \sum_{e \in E(T)} c(e)$, then the problem is trivially a Yes-instance. 
    Thus, in any non-trivial instance there is a polynomial dependence between $K$ and the sum of edge weights.}
\end{theorem}

\paragraph*{Polynomial-Time Turing Reductions}
We show \NP-hardness of our problem, under polynomial-time Turing reductions.
Such a reduction from some problem $\Pi$ to some problem $\Pi'$ solves problem $\Pi$ in polynomial time when given access to an oracle for problem $\Pi'$.
If problem $\Pi'$ were to have a polynomial-time algorithm, then this would yield a polynomial-time algorithm for problem $\Pi$ as well.
Conversely, if problem $\Pi$ is \NP-hard, then $\Pi'$ cannot be solved in polynomial time, unless $\P = \NP$.
\section{Construction}
\label{sec:construction}
We now describe the polynomial transformation which for a given instance $I = ((T,c), K)$ of \wts constructs in polynomial time an equivalent instance $I' = (G,B)$ of \tdp.
More precisely, the transformation additionally takes as input the treedepth $A$ of the wall $\wall_a$ where $a := K + 1 + c(T)$.
The size of $\wall_a$ is polynomial in the size of $I$ due to the unary encoding, however, no polynomial-time algorithm for the treedepth of walls is known.
We resolve this in our reduction by utilizing the oracle for planar treedepth to determine $A$ before applying the transformation given here.
This very step is what prevents us from giving a Karp reduction instead.

So let $I = ((T,c),K)$ and $A := \td(\wall_a)$ be given.
We encode the tree $T$ using a gadget for each vertex of the tree, and encode the edge weights by a corresponding number of connecting paths between the respective vertex gadgets.
The idea is that the operation of querying an edge $e$ at cost $c(e)$ is mapped to the operation of cutting each of the $c(e)$ connecting paths at the same total cost.
When all connecting paths are cut it remains to deconstruct the vertex gadget itself.
Accordingly, the transformed parameter $B$ is chosen as $K$ plus the treedepth of a vertex gadget.

For this approach to work, it needs to be expensive to ``cut through'' a vertex gadget. This ensures that an optimal treedepth decomposition begins by cutting the connecting paths rather than cutting through the vertex gadgets at their ends.
To achieve this, each vertex gadget contains as \emph{backbone} a large wall, providing the required connectivity and making a cut straight through the vertex gadget expensive.
In addition, we connect to this backbone a large number of \emph{anchors} of high treedepth which determine the exact value for the treedepth of the vertex gadget.
An example of the transformation is given in \cref{fig:transformation:example}.

\begin{figure}[ht]
    \centering
    \scalebox{0.9}{
        \constructionFigure
    }
    \caption{Transforming an instance of \wts to one of \tdp.}
    \label{fig:transformation:example}
\end{figure}

We now give a formal definition of the transformation.
For each vertex $v$ of $T$, the transformed graph $G$ contains a wall $\wall_v := \wall_{K + 1 + \deg_c(v)}$ where $\deg_c(v) := \sum_{u \in N_T(v)} c(uv)$ is defined as the sum of the incident edge weights at $v$.
Furthermore, for each $v$ the graph $G$ contains a set of $K + 1$ anchors.
Each anchor is a copy of the double wall $\wall_a^{(2)}$, with root vertex $r$ (see \cref{fig:preliminaries:doubling-graph,def:preliminaries:doubling-graph}).
Each anchor for $v$ is connected to $v$'s backbone by an edge from its root $r$ to a (distinct) terminal of the wall $\wall_v$.
Together, this defines the vertex gadget $G_v$.

It remains to connect the vertex gadgets according to the tree $T$.
For each edge $uv \in E(T)$ the graph $G$ contains $c(uv)$ connecting \emph{strand} vertices.
Each strand vertex is adjacent to one terminal of $\wall_u$ and one terminal of $\wall_v$.
More precisely, the $c(uv)$ strand vertices are connected to a contiguous interval of $c(uv)$ terminals of $\wall_u$ and $\wall_v$, in opposite orders.
This allows the nested drawing of strand vertices shown in \cref{fig:transformation:example}.
Since each backbone $\wall_v$ is a wall of dimension $K + 1 + \deg_c(v)$, each terminal receives exactly one edge leaving the wall.
This finishes the construction of graph $G$.
Observe that $G$ is planar and has maximum degree 3.

We define the problem parameter $B := 2K + 1 + A$, i.e., the instance asks whether $\td(G) \leq B$.
Recall that $A$ was defined as the treedepth of the wall used in each anchor.

\begin{lemma}
    The transformation from a \wts instance $I = ((T,c),K)$ and the value $A := \td(\wall_a)$ for $a := K + 1 + c(T)$, to a \tdp instance $I' = (G,B)$ can be performed in time polynomial in $|I|$.
\end{lemma}

\begin{proof}
    Recall that instance $I$ encodes each edge weight and the parameter $K$ in unary, so in particular, it has size at least $c(T) + K$.
    By construction, each backbone $\wall_v$ has dimension at most $a = K + 1 + c(T)$ so in particular, contains at most $2 \cdot (K + 1 + c(T))^2$ vertices, polynomial in $|I|$.
    Furthermore, the graph $G$ contains $|V(T)| \cdot (K + 1)$ anchors, each consisting of two walls $\wall_a$ plus a connecting vertex, all of which is polynomial in $|I|$.
    Finally, the graph $G$ contains $c(T)$ strand vertices, polynomial in $|I|$.
    Thus, the number of vertices (and edges) of $G$ is polynomial in $|I|$ and the graph can be constructed in polynomial time.

    By definition, the parameter $B := 2K + 1 + A$ is also polynomial in size (even when encoded in unary), and it can be computed in polynomial time \emph{if} the treedepth $A$ of the wall $\wall_a$ is given as an input.
\end{proof}
\section{Hardness Proof}
\label{sec:proof}
In this section we show \NP-hardness of \tdp by a reduction from \wts.
We use the transformation described in \cref{sec:construction} and show the following two statements.

\begin{restatable}{lemma}{upperBound}
\label{lem:equivalence:upper-bound}
    Let $I = ((T,c),K)$ be an instance of $\wts$ and $I' = (G,B)$ the transformed instance of \tdp.
    If $f(T,c) \leq K$ then $\td(G) \leq B$, i.e., if $I$ is a Yes-instance, then so is $I'$.
\end{restatable}
\begin{restatable}{corollary}{lowerBound}
\label{cor:equivalence:lower-bound}
    Let $I = ((T,c),K)$ be an instance of $\wts$ and $I' = (G,B)$ the transformed instance of \tdp.
    If $f(T,c) \geq K + 1$ then $\td(G) \geq B + 1$, i.e., if $I$ is a No-instance, then so is $I'$.
\end{restatable}

Before proving \cref{lem:equivalence:upper-bound,cor:equivalence:lower-bound}, we use them to show our main result.
\hardnessTheorem*

\begin{proof}
    Assume we are given an oracle that decides \tdp.
    Using this oracle we decide \wts in polynomial time as follows:
    Let an instance $((T,c), K)$ of \wts be given.
    If $K \geq c(T)$ or $K < c(e)$ for some edge $e \in E(T)$, answer Yes or No respectively.
    Otherwise, determine $A := \td(\wall_a)$ by a binary search over $[a,4a]$ using $\mathcal{O}(\log(K + c(T)))$ oracle queries.
    Note for this that $\wall_a$ is planar and has maximum degree 3, so it is a valid oracle input.
    Then, construct in polynomial time from the \wts instance $((T,c), K)$ and the treedepth value $A = \td(\wall_a)$, the instance $(G,B)$ with anchor $\wall_a^{(2)}$ for $a := K + 1 + c(T)$ as described in \cref{sec:construction}.
    Finally, decide $(G, B)$ using one more oracle call.
    By \cref{lem:equivalence:upper-bound,cor:equivalence:lower-bound}, this instance is equivalent to the original \wts instance.
    Thus, we have decided in polynomial time using only a logarithmic number of calls to an oracle for \tdp, an instance of \wts.
    This shows the claim.
\end{proof}

It remains to prove \cref{lem:equivalence:upper-bound,cor:equivalence:lower-bound}.
We begin with \cref{lem:equivalence:upper-bound}, giving an upper bound for the treedepth of the transformed graph $G$.
To this end, we first give an upper bound on the treedepth of a vertex gadget.

\begin{lemma}
\label{lem:equivalence:td-of-vertex-gadget}
    Let $I = ((T,c),K)$ be an instance of $\wts$ and $I' = (G,B)$ the transformed instance of \tdp.
    Then for each vertex $v \in V(T)$ it holds that the vertex gadget $G_v$ has treedepth at most $K + 1 + A$, where $A := \td(\wall_a)$ for $a := K + 1 + c(T)$.
\end{lemma}

\begin{proof}
    We give a deletion sequence for $G_v$.
    Recall that $G_v$ consists of a wall $\wall_v$ and some $K + 1$ anchors $\wall_a^{(2)}$ attached to $\wall_v$.
    Delete in $K + 1$ steps the root of each anchor.
    What remains are $2K + 2$ disjoint copies of $\wall_a$, two from each anchor, and the backbone $\wall_v$ which is a subgraph of $\wall_a$ since $\deg_c(v) \leq c(T)$.
    Thus, the resulting graph has treedepth $A := \td(\wall_a)$ and the full vertex gadget $G_v$ has treedepth at most $K + 1 + A$.
\end{proof}

Using this upper bound, the treedepth of the full transformed graph is bounded from above by the number of deletion rounds required to isolate each vertex gadget, plus this upper bound of $K + 1 + A$ to deconstruct each vertex gadget.

\upperBound*
\begin{proof}
    As before, define $A := \td(W_a)$ for $a := K + 1 + c(T)$, and recall the definition of $B := 2K + 1 + A$.
    By \cref{lem:equivalence:td-of-vertex-gadget}, in $G$ each vertex gadget $G_v$ has treedepth at most $K + 1 + A$.
    It now suffices to reduce the graph $G$ in $K$ deletion rounds to a disjoint union of vertex gadgets.
    For this, let $D$ be a decision tree for $(T,c)$, of cost at most $K$.
    Construct the top levels of an elimination tree $F$ for $G$ by replacing each node of $D$ belonging to some edge $e$ by a path on the corresponding $c(e)$ strand vertices.
    Finally, replace each leaf node $v$ of $D$ by an elimination tree for $G_v$ of height at most $K + 1 + A$ which exists by \cref{lem:equivalence:td-of-vertex-gadget}.
    Then $F$ has height at most $K + K + 1 + A = B$.
    Also, $F$ is a valid elimination tree for $G$.
    This can be seen by induction.
    If $T$ has just a single vertex $v$, then $D$ contains just one isolated node labeled $v$.
    Then, by construction $F$ is a valid elimination tree of $G_v$.
    Otherwise, the root of $D$ is labeled by some edge $e$ of $T$.
    Then $F$ starts with a path of strand vertices belonging to $e$.
    Deleting the last of these strand vertices splits the remaining subgraph of $G$ into two components, one for each component of $D - e$.
    By induction, for each of them our construction yields a valid elimination tree, both of which are connected to the final strand vertex.
\end{proof}

We now show that if $f(T,c) \geq K + 1$ then $\td(G) \geq B + 1$.
We first give some intuition for why the treedepth of $G$ is large in this case.
For this, recall that $G$ contains many anchors of high treedepth, connected within a vertex gadget via the vertex's backbone, and between different vertex gadgets via the connecting strand vertices.
Intuitively, an optimal elimination tree for $G$ first separates each anchor from the corresponding backbone, then deconstructs each anchor in parallel.
For the first step there are two possibilities.
One is to mirror a decision tree of $(T,c)$ to isolate each vertex, followed by separating each anchor from its vertex backbone in parallel.
The other is to separate each anchor from its vertex's backbone \emph{before} every vertex gadget is isolated, i.e., while some vertex gadgets $G_u$ and $G_v$ are still connected to each other.

The first approach takes $f(T,c)$ levels of the elimination tree to isolate each vertex, and another $K + 1$ levels to separate each anchor.
The second approach takes $2K + 2$ levels for the combined $2K + 2$ anchors of $G_u$ and $G_v$ plus possibly some more levels for the rest of the graph.
In either approach, deconstructing the final anchor takes another $A$ levels, for a total elimination tree height of $f(T,c) + K + 1 + A$, or $2K + 2 + A$.
If $f(T,c) \geq K + 1$, then either option yields an elimination tree of height greater than $B := 2K + 1 + A$.

The challenge in making this rigorous through an inductive proof is that the elimination tree might interleave strand and anchor deletions.
After some anchors are deleted/separated, the resulting graph $G'$ is no longer of the form of a transformed graph.

To account for this, we argue the lower bound over a more general family of graphs, which contains $G$, but which also permits graphs whose vertex gadgets have varying numbers of anchors.
In particular, we add for the pair $(T,c)$ a function $h: V(T) \rightarrow \mathbb{N}_0$ where $h(v)$ keeps track of the number of anchors connected to the backbone $\wall_v$ of vertex $v$.
We then call a graph $H$ a $(T,c,h,a)$-graph if it consists of vertex gadgets connected by strands as prescribed by $(T,c)$ and where the backbone $\wall_v$ of each vertex $v$ is connected to $h(v)$ anchors $\wall_a^{(2)}$.
More specifically, $\wall_v \cong \wall_{h(v) + \deg_c(v)}$ and each terminal of $\wall_v$ is connected to exactly one anchor or strand vertex.
When $T = \{v\}$ and $h(v) = 0$, then the $(T,c,h,a)$-graph contains no vertices.
Initially in our transformation, the graph $G$ is a $(T,c,h_0,a)$-graph where $h_0 \equiv K + 1$.

For each such graph we argue about the number of deletions required to cut off each vertex's anchors.
The two approaches then give two different lower bounds.
Isolating all vertices followed by separating each anchor requires at least $f(T,c) + \mu(h)$ levels, where $\mu(h) := \min_{v \in V(T)} h(v)$.
Alternatively, separating each anchor before all vertices are isolated requires at least $\nu(h)$ levels, where $\nu(h) := \min_{u \neq v \in V(T)} h(u) + h(v)$.
For the case that $T$ contains only one vertex (and so the domain of $h$ has only one element), we define $\nu(h) := \infty$.
Based on this, we define a potential function $\Psi$ which assigns to each $(T,c,h,a)$-graph $H$ a potential $\Psi(T,c,h) := \min\{f(T,c) + \mu(h), \nu(h)\}$ which gives a lower bound for the number of deletions required to separate each anchor.

The argument then runs that deleting the root $x$ of an optimal elimination tree for $H$ reduces the treedepth by exactly one, while reducing the potential by at most 1.
In general, $H - x$ does not have the correct form to be a $(T',c',h',a)$-graph, however, there always exists a minor $H'$ of $H - x$ of the correct form, and which has potential $\Psi(T',c',h') \geq \Psi(T,c,h) - 1$.
Since treedepth is minor-monotone, this shows that $\td(H) = \td(H - x) + 1 \geq \td(H') + 1 \geq \Psi(T',c',h') + A + 1 \geq \Psi(T,c,h) + A$ where the second-last step is by induction.

This leads to the following theorem which we revisit after proving the claimed properties of the potential.

\begin{restatable}{theorem}{tdFromPotential}
\label{thm:equivalence:td_from_potential}
    Let $H$ be a $(T,c,h,a)$-graph with $\Psi(T,c,h) \geq 1$.
    Then $$\td(H) \geq \Psi(T,c,h) + A$$
    where $A := \td(\wall_a).$
\end{restatable}

Before proving \cref{thm:equivalence:td_from_potential}, we show its usage in proving as corollary the desired lower bound claim.

\lowerBound*
\begin{proof}
    The reduction graph $G$ is a $(T,c, h, a)$-graph for $h \equiv K + 1$ and $a := K + 1 + c(T)$, with $A := \td(\wall_a)$.
    By \cref{thm:equivalence:td_from_potential} it follows that $\td(G) \geq \Psi(T,c, h) + A = \min\{f(T,c) + \mu(h), \nu(h)\} + A \geq \min\{(K + 1) + K + 1, 2K + 2\} + A = 2K + 2 + A = B + 1$.
\end{proof}

We now prepare the induction step.
We first describe in \cref{lem:equivalence:minor-of-h} the different graphs that can result from deleting a single vertex of a $(T,c,h,a)$-graph, then, in \cref{lem:equivalence:potential-of-h-prime}, argue that in each case the potential drops at most by 1.

\begin{lemma}
\label{lem:equivalence:minor-of-h}
    Let $H$ be a $(T,c,h,a)$-graph with $\Psi(T,c,h) \geq 2$ and let $x \in V(H)$.
    Then $H - x$ contains as a minor a graph $H'$ which is one of the following:
    \begin{enumerate}
        \item A $(T,c,h',a)$-graph with $0 \leq h'(v) = h(v) - 1$, for some $v \in V(T)$, and $h'(u) = h(u),$ for all $u \neq v$.
        \item A $(T,c',h,a)$-graph with $1 \leq c'(e) = c(e) - 1$, for some $e \in E(T)$, and $c'(e') = c(e'),$ for all $e' \neq e$.
        \item The disjoint union of a $(T_1,c|_{T_1}, h|_{T_1},a)$- and a (possibly empty) $(T_2, c|_{T_2}, h|_{T_2},a)$-graph, for $T_1 \cup T_2 = T-e$ where $c(e) = 1$.
    \end{enumerate}
\end{lemma}

\begin{proof}
    Vertex $x$ is either contained in some anchor, or is a strand vertex, or is contained in some vertex backbone $\wall_v$.
    In all three cases the approach is the same:
    Delete the corresponding anchor or strand, and shrink the vertex backbone(s), using \cref{obs:prel:wall-contract}.

    Consider first the case where $x$ is vertex $(i,j)$ of some vertex backbone $\wall_v$, i.e., contained in that wall's path $P_i$, and that the terminal of $P_i$ is connected to some anchor for vertex $v$.
    In this case, define as $H'$ the graph obtained from $H$ by deleting the connected anchor and contracting $\wall_v - x$ according to \cref{obs:prel:wall-contract}.
    The resulting graph is of type 1:
    Its $h$-function is decremented by 1 for vertex $v$ and unchanged otherwise, and the contracted wall $\wall'_v$ retains all of its other terminals and through them is connected to the remaining anchors and to the strand vertices for vertex $v$.

    The case where $x$ is contained in some anchor is analogous.
    Obtain $H'$ from $H$ by deleting said anchor and the terminal $x'$ it was connected to, followed by contracting $\wall_v - x'$ according to \cref{obs:prel:wall-contract}.
    Once more, $H'$ is of type 1.

    The remaining two possibilities are that $x$ is a strand vertex, or a backbone vertex whose wall-path is connected to a strand vertex.
    These two cases work by the same idea.
    In case $x$ is a strand vertex corresponding to some $T$-edge $uv$, define $x'$ and $x''$ as the connecting terminals of backbones $\wall_u$ and $\wall_v$.
    Then, define $H'$ as the graph obtained from $H$ by deleting $x, x',x''$ and contracting $\wall_u - x'$ and $\wall_v - x''$ according to \cref{obs:prel:wall-contract}.
    See \cref{fig:strand-deletion} for an example.
    \begin{figure}
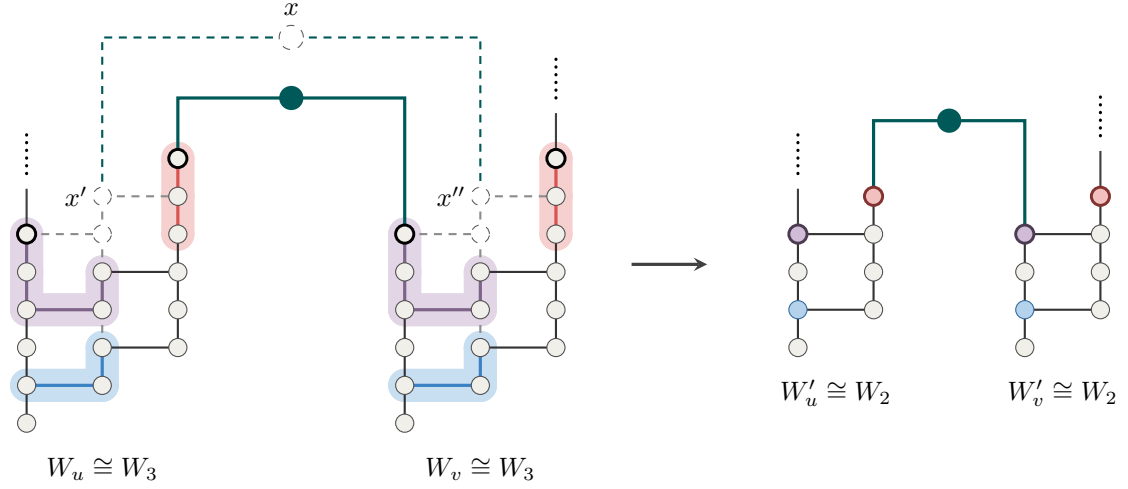

        \centering
        \strandCaseFigure
        \caption{The modification of graph $H$ in \cref{lem:equivalence:minor-of-h}, when deleting a strand vertex. The terminals $x'$ and $x''$ are deleted and the branch sets contracted. The purple branch set is the special case of the two merged branch sets of \cref{obs:prel:wall-contract}.}
        \label{fig:strand-deletion}
    \end{figure}
    Similarly, if $x$ is contained in backbone $\wall_v$ and its wall-path is connected to some strand vertex $x'$ corresponding to $T$-edge $uv$, define $x''$ as the terminal in $\wall_u$ connected to strand vertex $x'$.
    Then, define $H'$ as the graph obtained from $H$ by deleting $x,x',x''$ and contracting $\wall_u - x''$ and $\wall_v - x$.
    In each case, if $c(uv) \geq 2$, then the resulting graph $H'$ is connected and of type 2.
    If instead $c(uv) = 1$, then there are two cases.
    One case is that both $\wall_u$ and $\wall_v$ have dimensions at least $2 \times 2$ in $H$.
    In this case $H'$ contains two connected components, one for each component of $T - uv$, and is of type 3.
    The other case is that one of $\wall_u$ or $\wall_v$, without loss of generality $\wall_v$, is the $1 \times 1$-wall.
    This is the case when $v$ is a leaf in $T$ and $h(v) = 0$.
    Note that $\Psi(T,c,h) \geq 2$ ensures that this is not the case for both $u$ and $v$.
    In this (degenerate) case, the tree $T_2$ consists of the single vertex $v$ and $h|_{T_2}(v) = 0$.
    Then, $H'$ is of type 3 and the $(T_2, c|_{T_2}, h|_{T_2},a)$-graph is empty.
\end{proof}

\begin{lemma}
\label{lem:equivalence:potential-of-h-prime}
    Let $(T,c)$ be a weighted tree and $h: V(T) \rightarrow \mathbb{N}_0$ a function, and define $\Psi_0 := \Psi(T,c,h).$
    Then it holds
    \begin{enumerate}
        \item $\Psi(T,c,h') \geq \Psi_0 - 1$ when $0 \leq h'(v) = h(v) - 1$, for some $v \in V(T)$, and $h'(u) = h(u),$ for all $u \neq v$.
        \label{lem:potential:case-m}
        
        \item $\Psi(T,c',h) \geq \Psi_0 - 1$ when $1 \leq c'(e) = c(e) - 1$, for some $e \in E(T)$, and $c'(e') = c(e'),$ for all $e' \neq e$.
        \label{lem:potential:case-c}
        
        \item $\max\{\Psi(T_1,c|_{T_1}, h|_{T_1}), \Psi(T_2, c|_{T_2}, h|_{T_2})\} \geq \Psi_0 - 1$, for $T_1 \cup T_2 = T - e$ where $c(e) = 1$.
        \label{lem:potential:case-split}
    \end{enumerate}
\end{lemma}

\begin{proof}
    Recall that $\Psi(T,c,h) := \min\{f(T,c) + \mu(h), \nu(h)\}$.

    For case \ref{lem:potential:case-m}, note that $f(T,c)$ is unchanged, while $\mu(h') \geq \mu(h) - 1$ and $\nu(h') \geq \nu(h) - 1$ by definition of $h'$ and $\mu, \nu$.
    Thus, $\Psi(T,c,h') = \min\{f(T,c) + \mu(h'), \nu(h')\} \geq \min\{f(T,c) + \mu(h) - 1, \nu(h) - 1\} = \min\{f(T,c) + \mu(h), \nu(h)\} - 1 = \Psi_0 - 1$.

    For case \ref{lem:potential:case-c}, note%
    \footnote{Let $D'$ be a decision tree for $(T,c')$ with cost $f(T,c')$. Then it is also a decision tree for $(T,c)$ with cost at most $f(T,c') + 1$.}
    that $f(T,c') \geq f(T,c) - 1$.
    Thus, $\Psi(T,c',h) = \min\{f(T,c') + \mu(h), \nu(h)\} \geq \min\{f(T,c) - 1 + \mu(h), \nu(h) - 1\} = \min\{f(T,c) + \mu(h), \nu(h)\} - 1 = \Psi_0 - 1$.

    Case \ref{lem:potential:case-split} corresponds to a decision tree rooted at edge $e$ of weight $c(e) = 1$.
    It holds
    by definition of the search cost function that
    $f(T,c) \leq 1 + \max\{f(T_1, c|_{T_1}), f(T_2, c|_{T_2})\}$.
    Also, note that $\mu(h|_{T_i}) \geq \mu(h)$ and $\nu(h|_{T_i}) \geq \nu(h)$ since both $\mu$ and $\nu$ are minima, taken over a restricted domain. 
    We now compute
    \begin{equation*} \label{eq1}
      \begin{split}
        & \max\{\Psi(T_1,c|_{T_1}, h|_{T_1}), \Psi(T_2, c|_{T_2}, h|_{T_2})\} \\
        & = \max\Big\{ \min\{f(T_1, c|_{T_1}) + \mu(h|_{T_1}), \nu(h|_{T_1})\}, \min\{f(T_2, c|_{T_2}) + \mu(h|_{T_2}), \nu(h|_{T_2})\} \Big\} \\
        & \geq \max\Big\{ \min\{f(T_1, c|_{T_1}) + \mu(h), \nu(h)\}, \min\{f(T_2, c|_{T_2}) + \mu(h), \nu(h)\} \Big\} \\
        & = \min\Big\{ \max\{f(T_1, c|_{T_1}), f(T_2, c|_{T_2})\} + \mu(h), \nu(h)\Big\} \\
        & \geq \min\Big\{f(T,c) - 1 + \mu(h), \nu(h) - 1\Big\} \\
        & = \Psi_0 - 1.
      \end{split}
    \end{equation*}
    which shows the claim.
\end{proof}

To prove \cref{thm:equivalence:td_from_potential} by induction, we now combine the statements of \cref{lem:equivalence:minor-of-h,lem:equivalence:potential-of-h-prime}.
\begin{corollary}
\label{cor:induction-step}
    Let $H$ be a $(T,c,h,a)$-graph with $\Psi(T,c,h) \geq 2$.
    Then for every vertex $x \in V(H)$ it holds that $H - x$ contains as a minor a $(T',c',h',a)$-graph $H'$ with potential $\Psi(T',c',h') \geq \Psi(T,c,h) - 1$ and for which $c'(T') + ||h'|| \leq c(T) + ||h|| - 1$, where $||h|| := \sum_{v \in V(T)}h(v)$.
\end{corollary}

\begin{proof}
    Let $H$ and $x \in V(H)$ be given.
    Let $H''$ be the minor of $H - x$ given by \cref{lem:equivalence:minor-of-h}.
    If $H''$ is of type 1 or 2, define $H' := H''$.
    Then the claim holds by the first two cases of \cref{lem:equivalence:potential-of-h-prime}.
    If $H''$ is of type 3, then it is the disjoint union of a $(T_1, c|_{T_1}, h|_{T_1},a)$-graph $H_1$ and a $(T_2, c|_{T_2}, h|_{T_2},a)$-graph $H_2$.
    Define $H'$ as the $H_i$ with larger potential.
    Then, the claim holds by the third case of \cref{lem:equivalence:potential-of-h-prime}.
    In each case the bound on $c'(T') + ||h'||$ holds by definition of $H'$.
\end{proof}

\tdFromPotential*

\begin{proof}
    We prove the statement by induction over $c(T) + ||h||$, where $||h|| := \sum_{v \in V(T)} h(v)$.
    Observe that from $\Psi(T,c,h) \geq 1$ it follows that there exists a vertex $v \in V(T)$ with $h(v) \geq 1$.
    To see this, consider first the case that $T$ contains more than one vertex.
    Then $1 \leq \Psi(T,c,h) \leq \nu(h)$ which shows the claim.
    The other case is that $T$ contains only a single vertex $v$.
    Then $\nu(h) = \infty$ and so $1 \leq \Psi(T,c,h) = f(T,c) + \mu(h)$.
    However, $f(T,c) = 0$ in this case, since $T$ has no edge, and so it follows that $h(v) \geq 1$.

    From this, the case of $\Psi(T,c,h) = 1$ follows immediately.
    In this case $H$ contains an anchor attached to some $\wall_v$, thus $\td(H) \geq \td(\wall_a^{(2)})$.
    By \cref{obs:preliminaries:td-of-g-G}, the treedepth of the anchor $\wall_a^{(2)}$ is $\td(\wall_a) + 1 = A + 1$ and so it follows that $\td(H) \geq \td(\wall_a) + 1 = A + 1$.

    Now let $\Psi(T,c,h) \geq 2$.
    Note that $H$ is connected, and let $x$ be a vertex that attains the minimum in \cref{def:treedepth}.
    By \cref{cor:induction-step}, $H - x$ contains as a minor some $(T',c',h',a)$-graph $H'$ with $c'(T') + ||h'|| \leq c(T) + ||h|| - 1$ and with potential $\Psi(T', c', h') \geq \Psi(T,c,h) - 1 \geq 1$.
    By induction, $\td(H') \geq \Psi(T', c', h') + A \geq \Psi(T,c,h) - 1 + A$.
    It follows that $\td(H) = \td(H - x) + 1 \geq \td(H') + 1 \geq \Psi(T,c,h) - 1 + A + 1 = \Psi(T,c,h) + A$ which shows the claim.
\end{proof}

\newtheorem{question}{Question}

\section{Conclusion}
\label{sec:conclusion}
We showed that \tdp is \NP-hard under polynomial-time Turing reductions.
The degree bound is best possible, since graphs with maximum degree 2 are unions of paths and cycles, for which treedepth can be decided efficiently.
Nevertheless, some questions remain.
\begin{question}
    Is there an algorithm which for any $k$ computes the treedepth of the wall $\wall_k$ in time polynomial in $k$?
\end{question}

Such an algorithm would make our reduction a Karp reduction.
Note that the anchor-half is not required to be a wall.
If some different family of graphs with the required properties can be constructed, for which the treedepth can be determined exactly, this too would improve our reduction to a Karp reduction.
\begin{question}
    Is there an algorithm which constructs for each $k$ in time polynomial in $k$, some planar graph (ideally of maximum degree 3) whose size is polynomial in $k$ and whose treedepth is exactly $k$?
\end{question}

Beyond treedepth there is the well-known open problem of \emph{treewidth} on planar graphs.
\begin{question}
    Does there exist a polynomial-time algorithm which decides the treewidth of a planar graph?
    Is the problem \NP-hard instead?
\end{question}
It appears to be more difficult to prove \NP-hardness for planar treewidth, in part, because branchwidth can be computed in polynomial time on planar graphs~\cite{seymour1994call}.
Since branchwidth and treewidth differ by at most a factor of $\frac{3}{2}$, any hardness reduction would need to work around this fact.
Conversely, it may be the case that planar treewidth is in \P.

\bibliography{bibliography}

\end{document}